%% file: stochastic_geometry_Combes_Altman.tex
\documentclass[10pt,conference,letterpaper,twocolumn]{IEEEtran}
\makeatletter
\def\ps@headings{%
\def\@oddhead{\mbox{}\scriptsize\rightmark \hfil \thepage}%
\def\@evenhead{\scriptsize\thepage \hfil \leftmark\mbox{}}%
\def\@oddfoot{}%
\def\@evenfoot{}}
\makeatother
\usepackage[utf8]{inputenc} 

\usepackage{richard_config}

\begin{document}
\title{Flow-level performance of random wireless networks}
\author{Richard Combes and Eitan Altman}

\author{
\IEEEauthorblockN{Richard Combes (\IEEEauthorrefmark{1},\IEEEauthorrefmark{2}), and Eitan Altman (\IEEEauthorrefmark{2})}

\IEEEauthorblockA{\IEEEauthorrefmark{1}KTH, Royal Institute of Technology, Stockholm, Sweden}

\IEEEauthorblockA{\IEEEauthorrefmark{2}INRIA Sophia Antipolis, France}
}

\date{\today}
\maketitle
\begin{abstract}
    We study the flow-level performance of random wireless networks. The locations of base stations (BSs) follow a Poisson point process. The number and positions of active users are dynamic. We associate a queue to each BS. The performance and stability of a BS depend on its load. In some cases, the full distribution of the load can be derived. Otherwise we derive formulas for the first and second moments. Networks on the line and on the plane are considered. Our model is generic enough to include features of recent wireless networks such as 4G (LTE) networks. In dense networks, we show that the inter-cell interference power becomes normally distributed, simplifying many computations. Numerical experiments demonstrate that in cases of practical interest, the loads distribution can be well approximated by a gamma distribution with known mean and variance.
	 \footnote{Richard Combes is with KTH. Major parts of this work were done when he was in INRIA.}
\end{abstract}

\input{intro}

\input{model}

\input{traffic_models}

\input{line_networks}
\input{plane_networks}
\input{numerical_exp}
\input{conclusion}


\appendices
\input{proofs}
\end{document}

%% file: intro.tex
\section{Introduction}
	
	The most straightforward approach to performance evaluation of wireless networks is to assume that \acp{BS} locations are deterministic and follow a regular pattern such as an hexagonal grid on the plane. A more recent trend is to consider random \acp{BS} locations. Namely \acp{BS} locations form a stationary point process rather than a deterministic pattern. One of the motivations for this approach is that operational networks are hardly regular, due to irregular city planning, legal constraints on \ac{BS} placement and non homogeneity of traffic density.  When the \acp{BS} locations follow a \ac{PPP}, many performance indicators can be calculated in closed form or with simple numerical integrals, so that this approach is tractable. A summary of both the theory and networking applications can be found in \cite{Baccelli2009,Baccelli2009a}.
	
	Another characteristic of wireless network is the dynamic user behavior (also called \emph{flow-level dynamics}). Users enter the network at random locations and instants to receive a service, and leave after service completion. Users may move during their service. Network resources are shared between active users so that the time spent by different users is not independent i.e there is \emph{congestion}. Congestion due to dynamic user behavior is modeled by \emph{queueing theory}: each \ac{BS} is modeled as a queue, and the system performance is derived from its stationary distribution, providing that the queue is stable. Queuing models for wireless networks serving streaming traffic \cite{Kaufman1981,Roberts1981,BaccelliBartekKarray}, elastic traffic \cite{Bonald-dimensioning,Altman:2002:CMC:570645.570671} and a mix of both \cite{Bonald:2004:PBI:1012888.1005716} have been studied in the literature. An important characteristic is the performance of a \ac{BS} is its \emph{load}. In particular, without admission control, the number of active users grows without bound if the load is strictly larger than $1$.
	
	\underline{Related work}:
	
	Assuming that users are served by the closest \ac{BS}, the zones served by \acp{BS} are Poisson-Voronoi cells. The geometry of Poisson-Voronoi cells is a well studied topic because of its applications to physics and more recently networking. The mean and variance of the cell area is studied in \cite{Gilbert1962}. Statistical studies show that the cell area distribution can be approximated with good accuracy by a gamma distribution \cite{Kiang1966} or a generalized gamma distribution~\cite{Ferenc2007PhyA}. \cite{Calka2002} investigates the distribution of the smallest disk containing the Voronoi cell which can serve as an upper bound. The number of sides of the cell was studied in \cite{Calka2003sides}. Conditional to the number of faces of the cell, \cite{Zuyev1992} showed that the fundamental area of the cell follows a Gamma distribution and \cite{Calka2003} gives the distribution of the cell area. A more general gamma-type result is given in \cite{MollerZuyev1996}. In \cite{Foss1996}, the authors consider the distribution of the integral of a function of the distance over the cell, and calculate its first and second moments as well as its tail behavior. 
	
	There is a large amount of recent results on the performance evaluation of wireless networks using stochastic geometry. Most works assume that \acp{BS} locations follow a \ac{PPP}. Distribution of the covered areas is studied in \cite{BaccelliBT02}. The distribution of the \ac{SINR} and data rate in downlink cellular networks is derived in~\cite{AndrewsBaccelliGeometry}. \cite{Novlan2013} treats the uplink case. Downlink heterogenous networks are considered in~\cite{Singh2013,Dhillon2013,Dhillon2011,Novlan2012,Lin2012}, and the different types of \acp{BS} are assumed to follow independent \acp{PPP}. User association and fractional frequency reuse are studied in  \cite{Jo2011} and \cite{Novlan2012} respectively. The handover probability of a typical user is studied in \cite{decreusefond2012}. Concerning non-Poisson networks, \cite{Ganti2012} studies the distribution of shot-noise and \cite{Guo2013} examines the performance gap between \ac{PPP} and non-Poisson networks. 
	
	All these results concern ``single user'' performance evaluation and do not take into account dynamic user behavior. Namely, only the distribution of the \ac{SINR}, data rate etc. for a single user located at the origin are derived.  In \cite{Yu2011,Singh2013,Dhillon2013}, a random number of users is considered. However, the time a user spends in the network does not depend on his position or on the number of users connected to the same serving \ac{BS} i.e the congestion is not taken into account. For elastic traffic for instance, the time spent by a user to transmit a given amount of data depends on his \emph{throughput} which is a function of its position, interference and congestion. Furthermore~\cite{Yu2011,Singh2013} approximate the cell size distribution by a gamma distribution, as proposed by~\cite{Kiang1966,Ferenc2007PhyA}.

	\underline{Our contribution}: 
	
		In this paper, we combine stochastic geometry and queuing theory. We assume that \acp{BS} locations follow a \ac{PPP}, and we study the distribution of the \acp{BS} loads. The \ac{BS} load is expressed as the integral of a function over the area covered by this \ac{BS}, depending both on the distance to the \ac{BS} and the interference power created by interfering \acp{BS}. Hence the time spent in the network by a user can depend on its location, the interference and the number of users connected to the same \ac{BS}. We derive formulas to calculate characteristics of the loads distribution. We study both line and plane networks. In the case of line networks, in many cases of interest we can derive the full distribution of the load in terms of its Laplace/Fourier transform. For plane networks, we give formulas for the first and second moments of the loads. Hence we extend the results of~\cite{Foss1996} to take into account inter-cell interference. We do not use the gamma approximation of the cell size distribution as in~\cite{Yu2011,Singh2013}. Interference is modeled as \emph{Poisson shot-noise} and its distribution is well studied , e.g \cite{Lowen1990,Baccelli2009,Gulati2010}.	
	
	The rest of the paper is organized as follows: in section~\ref{sec:model}, we introduce the model considered and highlight the link between queuing theory and stochastic geometry. In section~\ref{sec:traffic_models}, we explain how our model describes many features of wireless networks, including types of traffic (voice, streaming, data etc.), frequency reuse, channel-aware scheduling and so on. In section~\ref{sec:line_networks} we consider \ac{PPP} networks on a line and show that in some cases the full distribution of the load can be obtained. We turn to \ac{PPP} networks on the plane in section~\ref{sec:plane_networks}, and calculate the first and second moments of the load. In section~\ref{sec:numerical_experiments} we perform numerical experiments to study the load distribution statistically. Section~\ref{sec:conclusion} concludes the paper.

%% file: model.tex
\section{The model}\label{sec:model}

	\subsection{The basic model}
	We consider~\acp{BS} placed in a Euclidian space $M$ equipped with the usual Euclidean norm $\norm{.}$. Typically we will consider $M = \RR$ for networks on a line and $M = \RR^2$ for networks in the plane. The \acp{BS} form a \ac{PPP} on $M$ with intensity $\lambda$ denoted $\Phi = \Set{x_n}_{n \in \ZZ}$. We denote by $\probap{.}$ and $\expecp{.}$ the Palm probability and expectation with respect to $\Phi$. See \cite{DaleyVereJones} for instance for the definition of Palm probability.
	Users are connected to the closest \ac{BS} (in distance) and we define ${\cal C}_n$ the zone served by \ac{BS} $n$:
	\eqs{  {\cal C}_n  = \Set{ z \in M :  \norm{z - x_n} \leq  \norm{z - x_{n^\prime}} , \forall n^\prime \neq n}.}
	We use the convention that $0 \in {\cal C}_0$. The set ${\cal C}_n$ is a Voronoi cell, and the collection $\Set{{\cal C}_n}_{n \in \ZZ}$ is known as the Poisson-Voronoi tessellation of $M$. We define ${\cal C}(0)$ the Voronoi cell of the point process $\Phi \cup \Set{0}$. By Slivnyak's theorem, ${\cal C}(0)$ and ${\cal C}_0$ have the same distribution under Palm probability. We call ${\cal C}(0)$ the typical cell. Consider marks $\Set{G_n}_{n \in \ZZ}$ in $\RR^{+}$, and a function $h: \RR^+ \to \RR^+$. 
	We define the \emph{shot-noise} at location $z$ by:
	\eqs{
		{\cal I}(z) = \sum_{n \in \ZZ} G_n h( \norm{z - x_n}). 
	}
	The random variable ${\cal I}(z)$ serves to model the inter-cell interference received at location $z$. The function $h$ models the signal attenuation due to distance. The marks $G_n$ model for instance shadowing and frequency reuse. We assume that the marks are \ac{i.i.d.} and that they are independent of $\Phi$. 	Results on the distribution of shot noise generated by a \ac{PPP} are recalled in appendix~\ref{sec:shotnoise}. By a slight abuse of notation we denote $G_0$ by $G$ when it does not create confusion. We define: 
	\eqs{ {\cal G}(s) = \expec{ \exp( - s G) },}
	the Fourier/Laplace transform of $G$. We will sometimes consider the path-loss to follow a power law (assumption~\ref{a:power_law_path}).
	\begin{assumptions}[Power law path-loss]\label{a:power_law_path}
	 $G \equiv 1$ and the path-loss is $h(r) = P r^{-\eta}$ with $\eta > 2$ and $P > 0$.
\end{assumptions}
	Finally we define the \emph{load} of the typical cell:
	\eq{ \label{eq:load}
	\rho_0 = \int_{{\cal C}(0)} f(z,{\cal I}(z)) dz, 
	}
	with $f:M \times \RR^+ \to \RR^+$ a positive measurable function. The quantity $f(z,{\cal I}(z)) dz$ denotes the infinitesimal load created by users whom enter the network at location $z$. \emph{The goal of this work is to study the Palm distribution of the load of the typical cell.} We consider two possible simplifying assumptions on $f$.
\begin{assumptions}[No interference]\label{a:no interf}
There is a positive function $f_0:\RR^+ \to \RR^+$ such that for all $z$: \eqs{ f( z, {\cal I}(z)) = f_0( \norm{z}).}
\end{assumptions}
\begin{assumptions}[Affine function of interference]\label{a:linear interf}
	There are positive functions $f_0:\RR^+ \to \RR^+, f_1:\RR^+ \to \RR^+,$ such that for all $z$:
	\eqs{ f( z, {\cal I}(z)) = f_0(\norm{z}) + f_1(\norm{z}) {\cal I}(z).}
\end{assumptions}
		The Palm distribution of the load is linked to the stationary distribution of the load by Slivnyak's \emph{inversion formula}:
 \eqs{ \expec{ F( \rho_0 ) } = \frac{ \expecp{  F(\rho_0) \abs{{\cal C}(0)}} }{ \expecp{\abs{{\cal C}(0)}}} =  \lambda  \expecp{  F(\rho_0) \abs{{\cal C}(0)}}, 
 }
	with $F$ any positive bounded function. We denote by ${\cal B}(z,\delta)$ the Euclidean ball centered at $z$ of radius $\delta$. For $N \geq 1$ and $(z_1,\cdots,z_N) \in M^N$ we define: \eqs{ {\cal B}_N(z_1,\cdots,z_N) = \cup_{i=1}^N {\cal B}(z_i, \norm{z_i}).} We denote by $B_N(z_1,\cdots,z_N)$ the Lebesgue measure of ${\cal B}_N(z_1,\cdots,z_N)$. 

	We will often consider the distribution of $\Phi$ conditional to the fact that a given set of points $\{ z_1,\cdots,z_N \}$ belong to the typical cell ${\cal C}(0)$. Conditional to the event $(z_1,\cdots,z_N) \in {\cal C}(0)^N$, under Palm probability, $\Phi$ is a \ac{PPP} on $M \setminus {\cal B}_N(z_1,\cdots,z_N)$ with intensity $\lambda$. Therefore, for $F$ a function of the point process $\Phi$ we define the expectation of $F$ conditional to the event that $\{ z_1,\cdots,z_N \}$ are in the typical cell:
	\als{ \expecp{ F(\Phi) | z_1,\cdots,z_N} &=  \expecp{ F(\Phi) | (z_1,\cdots,z_N) \in {\cal C}(0)^N} \sk 
	&=  \expecp{ F(\Phi \cap  ( \CC \setminus {\cal B}_N(z_1,\cdots,z_N) ) ) }. 	}

%% file: traffic_models.tex
\section{Propagation, Data Rate and Traffic models}\label{sec:traffic_models}
	
	The \emph{load} as defined in equation~\eqref{eq:load} appears in several models for the performance evaluation of wireless networks with flow-level dynamics. Namely, users enter the network at random times and locations, they are connected to the closest \ac{BS} and leave the network after receiving service from the \ac{BS}. The time spent by a user in the network depends on:
\begin{itemize}
	\item his location with respect to the serving \ac{BS} because the power of the received signal decreases with distance,
	\item his location with respect to the other \acp{BS} because the \acp{BS} interfere with each other,
	\item the number of active users and their positions because of congestion. Namely the radio resources are shared between active users at each \ac{BS}.
\end{itemize}
	We describe a few queuing models used for different types of traffic. Users enter the network according to a homogeneous \ac{PPP} of intensity $\lambda_{us}$ (not to be confused with $\lambda$). We will denote by $R(z)$ the data rate of a user located at $z$ i.e his throughput when he is the only active user served by the closest \ac{BS}.
	\subsection{Traffic models}
	\subsubsection{Voice traffic} The available bandwidth is divided into $C$ circuits (sometimes called subcarriers). When a user enters the network he is assigned a circuit. If there are no circuits available the user is blocked and leaves the network. A user stays an exponential time with mean $1/\mu$ in the network. Therefore, each \ac{BS} is a M/M/C/C queue. The load is:
\eqs{ 	\rho_0 = \int_{{\cal C}(0)} (\lambda_{us}/\mu) dz = (\lambda_{us}/\mu) \abs{ {\cal C}(0)  }.}
	The performance indicator of interest is the blocking rate (the proportion of blocked users), which is a function of the load, and is given by the Erlang B formula (see for instance \cite{kleinrock1976queueing1}).
	\subsubsection{Streaming traffic} For streaming traffic, the amount of resources allocated to a user depends on his data rate. A user stays an exponential time with mean $1/\mu$ in the network. The bandwidth is divided into $C$ circuits of equal size. Users must achieve a minimal data rate $R_{min}$ and are allocated $ \ceil{ \frac{C R_{min}}{R(z)}} $ circuits. If a user arrives and there are not enough circuits to serve him he is blocked. The load is: \eqs{   \rho_0 = \int_{{\cal C}(0)} \frac{\lambda_{us}}{C \mu}  \ceil{ \frac{ C R_{min} }{R(z)} } \indic_{[0,C]}\left( \frac{ R_{min} }{R(z)}  \right) dz.}
	This model is known as \emph{Multi-rate Erlang model} and the blocking rate can be calculated rapidly using the \emph{Kaufman-Roberts} algorithm (\cite{Kaufman1981,Roberts1981}). The blocking rate is not always an increasing function of the load (\cite{BonaldNetworks}), but the load can serve as a first-order performance measure.
	\subsubsection{Adaptive streaming traffic} For adaptive streaming traffic, users can adjust their instantaneous throughput based on the available bandwidth. Video-on-demand services usually follow this model, and the throughput is adjusted by changing the level of video encoding. A user stays an exponential time with mean $1/\mu$ in the network. Each \ac{BS} can be modeled as a M/M/$\infty$ queue with load $\rho_0 = (\lambda_{us}/\mu) \abs{ {\cal C}(0)} $.	The number of active users in stationary state is a Poisson random variable with mean $\rho_0$. Assume that the available bandwidth is shared equally among active users. Denote by $N$ a Poisson random variable with mean $\rho_0$. Then the expected throughput of a user at $z$ in steady state is:
	\eqs{
	R(z) \expec{ \frac{1}{N} \Big| N > 0 } = R(z) \frac{1 - e^{-\rho_0}(1 + \rho_0 )}{\rho_0(1 - e^{-\rho_0})}.
	}
\subsubsection{Elastic traffic} For elastic traffic, the bandwidth is shared fairly between active users as for adaptive streaming. Users download a random amount of data with expectation $\sigma$. Each \ac{BS} is an M/G/1 processor sharing queue (\cite{kleinrock1976queueing1,Bonald-dimensioning}) with load:
\eq{ \label{eq:elastic_loads}
	\rho_0 = \sigma \lambda_{us} \int_{{\cal C}(0)} \frac{1}{R(z)} dz .
}
In particular both the mean number of active users and the blocking rate are functions of the load (\cite{Bonald-dimensioning}). For instance the mean number of active users in stationary regime is $\rho_0/(1 - \rho_0)$. It is noted that for elastic traffic, the time spent by a user is inversely proportional to his data rate, so that cell edge users stay longer than cell center users.	
	\subsection{Data rate calculation}\label{subsec:influence_interference}
	\subsubsection{Spectral efficiency}
	We define the \ac{SINR} at location $z$:
	\eqs{ 
	S(z) = \frac{ h(z) }{ N_0 +  {\cal I}(z) },
	}
	with $N_0$ the thermal noise power. The user data rate at a given location depends on the interference because it is a function of the \ac{SINR}. Assuming that the channel between a \ac{BS} and a user is \ac{AWGN}, the data rate is given by the Shannon formula:
	\eq{ \label{eq:elastic_loads_data_rate}
	R(z) = w \log_2\left( 1 + S(z) \right), }
	with $w$ the bandwidth used for data transmission. As shown in \cite{MogensenLteCapacity} a modified Shannon formula $R(z) = \overline{w} \log_2(1 + S(z)/\overline{S})$ with $\overline{S} = 1.25$ and $\overline{w} = 0.75 w$ provides a very good approximation to the practical performance of LTE.
	
	For wide-band systems such as \ac{CDMA} where $w$ is large and the \ac{SINR} is low, the data rate can be well approximated by a \emph{linear} function of the \ac{SINR}:\eqs{
		R(z) \approx  w S(z) / \log(2), } and in this case, $1/R(z)$ is an \emph{affine} function of the interference ${\cal I}(z)$.
	\subsubsection{Path-loss, shadowing and fading}
	Distance-dependent path-loss is captured by the function $h$. Shadowing is captured by $G$, and typically $G$ can be taken as a log-normal variable. We can also account for Rayleigh fading of the useful signal by considering (\cite{Karray2009}):
	\als{ 
	R(z) &= w  \int_{0}^{+\infty} \log_2 \left( 1 + \xi S(z) \right) e^{-\xi} d\xi \sk
	&= w e^{1/S(z)} E_1(1/S(z))/\log(2), 
	}
	with $E_1(.)$ the exponential-integral function. The fast-fading of interfering signals has a relatively small impact because there is typically a large number of interfering \acp{BS} with independent fading processes (see for instance~\cite{CombesPeva2011}). 
	
	Similarly, \ac{MIMO} can be taken into account by changing the function mapping \ac{SINR} into the corresponding data rate.
	\subsubsection{Channel-aware scheduling}
	In downlink cellular networks serving elastic traffic, channel measurements are available at the \ac{BS} so that a scheduler picks the user with the best channel condition to transmit e.g a proportional fair scheduler. We can account for channel-aware scheduling in queuing models (e.g \cite{Bonald-dimensioning,Borst2005}). Define $R(z)$ the data rate of a user when his fading is maximal. In practice $R(z)$ is the data rate provided by the highest modulation and coding scheme. Then, for elastic traffic we can define the load as equation~\eqref{eq:elastic_loads}, and the stability condition is $\rho_0 < 1$. Each \ac{BS} can be represented by a M/G/1/PS queue with state-dependent service rate.
	
	\subsubsection{Frequency reuse schemes}\label{subsec:sfr}
	In interference limited networks, it can be beneficial that \acp{BS} do not transmit on the whole bandwidth at full power. We consider the following \emph{soft reuse} scheme. The bandwidth is split into $b > 1$ sub-bands of size $w/b$. Based on a threshold $r_{edge} \geq 0$ on the distance to the serving \ac{BS}, users are split into far users (``edge users'') and close users (``center users''). Each \ac{BS} chooses one of the sub-bands at random on which it transmits at full power and serves edge users. On the other sub-bands, the \ac{BS} serves center users and transmits at reduced power. The ratio between the reduced power and the maximal power is denoted $\kappa$. It was shown in~\cite{BonaldHegdeFrequencyReuse} that soft reuse noticeably increases the network capacity for elastic traffic. 

	Since choice of sub-bands by different \acp{BS} are independent, the interference received by a user is a shot noise with:
\eqs{
	\proba{ G = 1 } = 1/b  \;\;,\;\;  \proba{ G = \kappa } =	1 - 1/b.
}
By a thinning argument, the interference is also the sum of two shot noises generated by two \emph{independent} \acp{PPP} with intensity $\lambda/b$ and $\lambda(b-1)/b$ respectively.
	
	A \ac{BS} can be represented by two M/G/1/PS queues, one for the edge users and one for the center users with respective loads $\rho_0^{edge}$ and $\rho_0^{center}$:
	\als{
	 \rho_0^{edge} &= \sigma \lambda_{us} \int_{{\cal C}(0)} \frac{1}{R^{edge}(z) } \indic_{ [r_{edge},+\infty)}(\abs{z}) dz \sk
	 \rho_0^{center} &= \sigma \lambda_{us} \int_{{\cal C}(0)} \frac{1}{R^{center}(z)} \indic_{ [0,r_{edge}) }(\abs{z}) dz.
	}
The data rates are calculated by:
	\als{
	 R^{edge}(z) &=  (w/b) \log_2\left( 1 + S(z)  \right), \sk
	 R^{center}(z) &=  (w(b-1)/b) \log_2\left( 1 + \kappa S(z) \right). 
	} 
 We also consider a scheme called \emph{hard reuse} where a \ac{BS} chooses one of the sub-bands on which it transmits at full power and does not transmit on the others. This scheme can be seen as a particular case of soft reuse with $\kappa = 0$ and $r_{edge}=0$. Namely all users are considered as edge users.

\subsubsection{BS activity patterns}\label{subsec:sleep}

	In a realistic setting, \acp{BS} only transmit when they have users to serve, so that the inter-cell interference they create depends on their state. Namely neighboring \acp{BS} can be represented by coupled queues. For more than two \acp{BS}, this model is known to be intractable. ~\cite{Bonald04wirelessdata} proposes a ``fluid approximation'': \acp{BS} are modeled by independent queues, but the amount of time they transmit is equal to their load. The loads are obtained as a solution to a fixed point equation. While this is tractable for regular networks, in Poisson networks, even the expected load seems difficult to obtain, because the load of a~\ac{BS} depends on the load of its neighbors.
	
	We suggest to use the (tractable) approximation introduced by \cite{Dhillon2013}: the \acp{BS} are modeled by independent queues, and they are active with probability $p = \min( \expecp{\rho_0} , 1 )$. This is an independent thinning, so that the interference received at location $z$ is a shot noise generated by a \ac{PPP} with intensity $p \lambda$. The value of $\expecp{\rho_0}$ can then be found as a solution to a fixed point equation. 
	
	Furthermore, this issue is only significant in low loads, when the traffic demand is small compared to the network capacity. Hence, in operational networks, it should not be a major issue if we assume proper network dimensionning. 
	
	Also, the ``sleep mode'' functionality studied in~\cite{Tsilimantos2013} can be modeled the same way: \acp{BS} are switched off with probability $p$ resulting in an independent thinning.

%% file: line_networks.tex
\section{Line networks}\label{sec:line_networks}

	We first analyze line networks ($M = \RR$). Line networks are interesting because we can obtain the full distribution of the load in some cases. In this section $z \in \RR$ and $r \in \RR^+$.
\subsection{Typical cell}
	Under Palm probability, there is a \ac{BS} at $0$, that is $x_0 = 0$. We define $x_{l} = -x_{-1}$  and $x_{r} = x_{1}$, which are the neighbors of the central \ac{BS} on the left and right respectively. The typical cell is:	
	\eqs{ {\cal C} (0) = [-x_l/2,x_r/2]. } 
	The cell load is:
	\eq{\label{eq:line_load} \rho_0 = \int_{-x_l/2}^{x_r/2} f(z,{\cal I}(z))dz.}
	The geometry of the typical cell follows from the definition of a \ac{PPP}.
\begin{proposition}\label{prop:geometry_line_cell}
	$(x_l,x_r)$ are independent and both follow an exponential distribution with parameter $\lambda$.
\end{proposition}
\begin{proof}
	By definition of the \ac{PPP} , for $(X_l,X_r) \in (\RR^+)^2$ :
	\als{
	\probap{ x_l > X_l , x_r > X_r  } &= \probap{ \Phi \cap [-X_l,X_r] = \Set{0} } \sk
	&= \exp( - \lambda ( X_l + X_r )  ) \sk
	&= \exp( - \lambda X_l ) \exp( - \lambda X_r).
	}
\end{proof}
	This is the main reason why the distribution of the loads can be derived in some cases: the geometry of typical cell is described by two independent exponential variables.
\subsection{No interference}
	In this subsection we consider assumptions~\ref{a:no interf}, so that the load does not depend on interference. We derive its Laplace transform in Theorem~\ref{th:load_line_no_interf}. The proof of Theorem~\ref{th:load_line_no_interf} is given in appendix~\ref{subsec:proof_th_load_line_no_interf}. We define the auxiliary function $F_0$:
\eq{\label{eq:line_load_distance} F_0(r) = \int_0^r f_0(z)dz. }	
\begin{theorem}\label{th:load_line_no_interf}
	With assumptions~\ref{a:no interf}, the Laplace transform of the load is:
	\eqs{ \expecp{ \exp( -s \rho_0 )} = \left(  \lambda \int_{0}^{+\infty} \exp( -s F_0(r/2) - r \lambda) dr  \right)^2 }	
\end{theorem}
\begin{corollary}
	\begin{itemize}
	\item (i) If $f_0(z) = 1$, $\rho_0$ follows an Erlang distribution with parameters $(2,\lambda)$.
	\item (ii) If $f_0(z) = \abs{z}^\alpha$, $\alpha \geq 0$ , $\rho_0$ is distributed as the sum of two independent Weibull \acp{r.v.} with parameters $(  4\lambda)^{-\alpha-1}(\alpha+1)^{-1},  ( \alpha + 1)^{-1} )$.
	\end{itemize}
\end{corollary}
\begin{proof}
	(i) If $f_0(z) = 1$, then $F_0(r) = r$ so that 
\eqs{
	\rho_0 = \frac{x_l + x_r}{2}.
}
	From proposition~\ref{prop:geometry_line_cell}, $\rho_0$ is the sum of two independent exponentially distributed \acp{r.v.} with  parameter $2 \lambda$. Hence $\rho_0$ has an Erlang distribution with parameters $(2,\lambda)$.

	(ii) If $f_0(z) = \abs{z}^\alpha$, $\alpha \geq 0$, then $F_0(r) = \frac{r^{\alpha+1}}{\alpha+1}$ so that:
	\eqs{
	\rho_0 = \frac{(x_l/2)^{\alpha+1} + (x_r/2)^{\alpha+1}}{\alpha+1}.
	}
	$x_l/2$ is exponentially distributed with parameter $2 \lambda$, so that $(x_l/2)^{\alpha+1}$ follows a Weibull distribution with parameter $(4 \lambda)^{-\alpha-1}(\alpha+1)^{-1},  ( \alpha + 1)^{-1} )$. Hence $\rho_0$ is indeed a sum of two \ac{i.i.d.} Weibull \acp{r.v.} with parameters $( 4\lambda)^{-\alpha-1}(\alpha+1)^{-1},  ( \alpha + 1)^{-1} )$.
\end{proof}
	We observe that the cell size follows a Gamma distribution (Erlang is a particular case of Gamma). The Weibull distribution was introduced in \cite{Weibull1951}. The sum of \ac{i.i.d.} Weibull \acp{r.v.} was studied in \cite{Yilmaz2009}.
\subsection{Affine function of interference}\label{subsec:linear_interference_line}
	In this subsection we consider assumptions~\ref{a:linear interf}, so that the load is an affine function of interference. This is of interest for instance in wide-band systems serving elastic traffic as mentioned in subsection~\ref{subsec:influence_interference}. In this case as well we can derive the full distribution of the \ac{BS} load. 

	The \ac{BS} load can be written:
	\eq{\label{eq:load_line_linear_interf}
	 \rho_0 =  F_0(x_l/2) + F_0(x_r/2) +  \int_{-x_l/2}^{x_r/2} f_1(z) {\cal I}(z) dz.
	}
	
	We define two auxiliary functions ${\cal L}$ and $ {\cal K}$:
	\als{ {\cal L}(s, x_l , x_r) &= \expecp{ \exp \left\{ -s \int_{-x_l/2}^{x_r/2} f_1(z) {\cal I}(z) dz \right\} | x_l , x_r  } \sk
	&= {\cal G}( s {\cal K}(-x_l,x_l,x_r) ) {\cal G}( s {\cal K} (x_r,x_l,x_r) ) \sk 
	& \exp( -\lambda \int_{ \RR \setminus [-x_l,x_r] } (1 -  {\cal G}( s {\cal K} (z,x_l,x_r) ) ) dz ), }
	and:
	\eqs{ 
	\label{eq:th_1_def_k} {\cal K}(u,x_l,x_r) =  \int_{-x_l/2}^{x_r/2} f_1(z) h(\abs{z-u}) dz. 
	}
	The Laplace transform of the load is given by Theorem~\ref{th:linear_interf_plane}. The proof of Theorem~\ref{th:linear_interf_plane} is given in appendix~\ref{subsec:proof_th_linear}.
\begin{theorem}\label{th:linear_interf_plane}
	With assumptions~\ref{a:linear interf}, the Laplace transform of the \ac{BS} load is given by:
	\als{ & \expecp{ \exp( -s \rho_0) } = 2 \lambda^2 \int_{0}^{+\infty} \int_{0}^{r_2} {\cal L}(s,r_1 ,r_2) \sk 
	&\exp \left\{ -\lambda ( r_1  + r_2 ) - s ( F_0(r_1/2)  + F_0(r_2/2)) \right\} d r_1 d r_2.	}
\end{theorem}
	The moments of the load can be recovered from the Laplace transform. Corollary~\ref{cor:linear_interf_plane1} gives  a fairly simple expression for the expected load.
\begin{corollary}\label{cor:linear_interf_plane1}
	With assumptions~\ref{a:linear interf}, the expectation of the \ac{BS} load is:
	\eqs{
	  \expecp{\rho_0} =  2 \int_{0}^{+\infty} ( f_0(r) + 2 \lambda \expec{G} H(r)  f_1(r) ) e^{- 2 \lambda r} dr,
	}
	with $H(r) = \int_{r}^{+\infty} h(u) du$.
\end{corollary}
	If we further consider power law path loss, the expected load has a simple form given by Corollary~\ref{cor:linear_interf_plane2}.
\begin{corollary}\label{cor:linear_interf_plane2}
	With assumptions~\ref{a:power_law_path} and~\ref{a:linear interf}, the expectation of the \ac{BS} load is:
\eqs{
	  \expecp{\rho_0} =  2 \int_{0}^{+\infty} ( f_0(r) + 2 \lambda P r^{1 - \eta} f_1(r) / (\eta - 1)  )  e^{- 2 \lambda r} dr.
	}
\end{corollary}

%% file: plane_networks.tex
\section{Plane networks}\label{sec:plane_networks}
	We now turn to plane networks so that $M = \RR^2$. Analysis for plane networks is more involved than for line networks because the geometry of the typical cell is more complex. In this section we derive the first and second moments of the cell loads.
\subsection{Moments of the cell load}	
	In this section we identify $\RR^2$ with the complex plane $\CC$ to simplify notation. $z , z^\prime$ are locations in the complex plane, with $z = r e^{i \theta}$ , $z = r^\prime e^{i \theta^\prime}$ their polar coordinates (with $i$ the imaginary unit) and $x, x^\prime$ are locations in the complex plane used as integration variables. It is noted that $B_1(z) = \pi \abs{z}^2$ and:\eqs{B_2(z,z^{\prime}) = \pi ( \abs{z}^2 + \abs{z^{\prime}}^2) - \abs{z}^2 A(z^{\prime}/z)}
with $A(z)$ the area of the intersection of the unit disk, and a disk of center $z-1$ and radius $|z|$.
	
	The $N$-th moment of the cell load can be calculated as a $2 N$ dimensional integral. The moments might be infinite.
\begin{theorem}\label{th:plane_load_moments}
	The $N$-th (Palm) moment of the load is:
	\al{
		\expecp{ \rho_0^N } &= \int_{ \CC^N}  \expecp{ \prod_{i=1}^N f(z_i, {\cal I}(z_i) ) | z_1,\cdots,z_N}  \sk 
		&\exp( - \lambda B_N(z_1,\cdots,z_N) ) dz_1 \cdots dz_N.\label{eq:th:plane_load_moments}
	}
\end{theorem}
The Laplace transform of the interference $({\cal I}(z_1), \cdots, {\cal I}(z_N))$ conditional to $(z_1,\cdots,z_N) \in {\cal C}(0)^N$ is given by proposition~\ref{prop:shot_noise_laplace} stated in the appendix. Also, the expectation on the \ac{r.h.s.} of~\eqref{eq:th:plane_load_moments} can be calculated using the Plancherel theorem (proposition~\ref{prop:parseval}).

	While it is theoretically possible to calculate all the moments of the load using Theorem~\ref{th:plane_load_moments}, the amount of computing power required increases exponentially with $N$. Hence its practical use is limited to the first and second moments. 
\begin{theorem}\label{th:plane_load_mean}
	(i) The first (Palm) moment of the load is:
		\eq{ \label{eq:plane_load_mean} \expecp{ \rho_0 } = 2 \pi \int_{0}^{+\infty} r K(r) e^{-\lambda \pi r^2} dr, }
	with: $K(z) = \expec{f(z, {\cal I}(z)) | z}$.
\end{theorem}
The Laplace transform of ${\cal I}(z)$, conditional to $z \in {\cal C}(0)$ is given by proposition~\ref{prop:shot_noise_plane}.
\begin{proof}
	(i) We have that $  B_1(z_1) = \pi \abs{z_1}^2$, so specializing Theorem~\ref{th:plane_load_moments} to $N = 1$ gives:
	\eq{ \label{eq:thm3eq1}
		\expecp{ \rho_0 } = \int_{\CC} K(z) e^{ - \lambda \pi \abs{z}^2 } dz.
	}
	By circular symmetry, we have that $K(z) = K(\abs{z})$ for all $z$. Rewriting the integral~\eqref{eq:thm3eq1} in polar coordinates proves the result:
	\eqs{
		\expecp{ \rho_0 } = 2 \pi  \int_{0}^{+\infty} r K(r) e^{ - \lambda \pi r^2 } dr
	}
\end{proof}
The second moment of the cell load is given by theorem~\ref{th:plane_load_var}, and its proof is presented in appendix~\ref{subsec:proof:plane_load_var}.
\begin{theorem}\label{th:plane_load_var}
	a) The second (Palm) moment of the load is:
	\als{  \expecp{ \rho_0^2 } = 8 \pi \int_{0}^{+\infty} r \Lp \int_{r}^{+\infty} r^\prime \Lp \int_{0}^{\pi}   Q(r,r^\prime, \theta) d \theta \Rp d r^\prime\Rp dr.}
with:
	\eqs{ Q(r,r^\prime, \theta) = L(r,r^\prime e^{i\theta}) \exp \Lp - \lambda r^2 B_2(1,e^{i\theta} r^\prime/r) \Rp,}
	and:
	\eqs{ L(z,z^\prime) = \expec{   f(z, {\cal I}(z)) f(z^\prime, {\cal I}(z^\prime)) | z,z^\prime}.}
	
\end{theorem}
	The Laplace transform of $({\cal I}(z), {\cal I}(z^\prime))$  conditional to $(z,z^\prime) \in {\cal C}(0)^2$ is given by specialization of theorem~\ref{prop:shot_noise_laplace} to $N = 2$ and $\lambda(dx) = \lambda \indic_{ \CC \setminus {\cal B}_2(z,z^{\prime})}(x) dx$.

	Theorems~\ref{th:plane_load_mean} and~\ref{th:plane_load_var} are an extension of the results of \cite{Foss1996} to take into account the influence of interference. In order to compute $K$ and $L$, we must compute the \ac{p.d.f.} of the interference at various locations. 

\subsection{Affine function of interference}\label{subsec:linear_interference_plane}
 As done for plane networks in subsection~\ref{subsec:linear_interference_line}, we consider the case where the load at a location is an affine function of the interference (assumptions~\ref{a:linear interf}). In this case calculating the full distribution of the interference is not needed. Obtaining the first and second moments is sufficient which simplifies the calculations. 
\begin{corollary}\label{cor:plane_load_mean_linear}
	Under assumptions~\ref{a:linear interf}, the first (Palm) moment of the load is:
 \eqs{ 
 \expecp{ \rho_0 } = 2 \pi \int_{0}^{+\infty} r \Lp f_0(r) +  2 \pi \lambda \expec{G} H(r) f_1(r) \Rp e^{- \lambda \pi r^2} dr .
 }
 with $H(r) = \int_{r}^{+\infty} r h(r) dr$.
\end{corollary}
\begin{corollary}\label{cor:plane_load_mean_linear_path}
	Under assumptions~\ref{a:power_law_path} and~\ref{a:linear interf}, the first (Palm) moment of the load is:
	\eqs{ 
	 \expecp{ \rho_0 } = 2 \pi \int_{0}^{+\infty} r \Lp f_0(r) +  2 \pi \lambda P \frac{r^{2-\eta}}{\eta - 2} f_1(r) \Rp e^{ - \lambda \pi r^2 } dr .}
\end{corollary}
	Under assumptions~\ref{a:power_law_path} and~\ref{a:linear interf}, the $L$ function introduced in Theorem~\ref{th:plane_load_var} simplifies similarly.
\subsection{Elastic traffic in the low SINR regime}\label{subsec:elastic_low_sinr}
	When considering elastic traffic in the low \ac{SINR} regime, the expected load has a remarkably simple expression derived in corollary~\ref{cor:th:plane_load_var}. The equation for the load (the function $f$) for elastic traffic is explained in section~\ref{sec:traffic_models}. The proof of corollary~\ref{cor:th:plane_load_var} is found in appendix~\ref{subsec:cor:th:plane_load_var}.
\begin{corollary}\label{cor:th:plane_load_var}
	Consider function $f$ given by equations~\eqref{eq:elastic_loads} and ~\eqref{eq:elastic_loads_data_rate}, and assumption~\ref{a:power_law_path}. Then we have the lower bound on the expected load:
	\eq{\label{eq:load_bound}
		\expecp{ \rho_0 } \geq \frac{ \log(2) \lambda_{us} \sigma  }{w \lambda} \Lp   \frac{N_0}{ P ( \sqrt{\pi \lambda} )^{\eta}} \Gamma( 1 + \eta/2 )  +  \frac{2}{\eta - 2} \Rp.
	}
	In particular:
	\eqs{
		\expecp{ \rho_0 } \tends{\eta}{2^+} +\infty.	}
\end{corollary}
	It is noted that the lower bound~\eqref{eq:load_bound} is particularly simple and is \emph{tight in the low \ac{SINR} regime}. 
	Corollary~\ref{cor:th:plane_load_var} states that the average load becomes infinite for elastic traffic when the signal attenuation is a power law, with path-loss exponent close to free space propagation. Another interesting fact is that the average load does not depend on the transmitted power $P$ in the interference limited regime. Indeed, in stationary probability, the average distance between a user located at $0$ and the closest \ac{BS} is $1/\sqrt{\pi \lambda} = 2 \lambda \pi \int_{0}^{+\infty} r^2 e^{-\lambda \pi r^2}$. So the term $P ( \sqrt{\pi \lambda} )^{\eta} /N_0$ is the ratio between the received power at distance $1/\sqrt{\pi \lambda}$ and the thermal noise power $N_0$. When this term is large, the first term of~\eqref{eq:load_bound} becomes negligible, and 
	\eqs{
	\expecp{ \rho_0 } \approx \lambda_{us} \sigma \frac{2 \log(2)}{(\eta - 2)w \lambda}.
	}
	Also, in this regime, the load is inversely proportional to the \ac{BS} density. Namely, \emph{ the network capacity is proportional to the number of deployed \ac{BS} .} 
\subsection{Dense networks}
	In dense networks ($\lambda \to +\infty$), the interference becomes normally distributed. Then $K$ and $L$ can be calculated as integrals of the Gaussian distribution, and we only need to calculate the first and second moments of the interference. In the case of a \ac{PPP} on the line (i.e $M = \RR$), convergence of shot noise to a Gaussian process is well known (\cite{davenport1958introduction,feller1966introduction, Lowen1990}). Theorem~\ref{th:dense_network_gen1} provides a generalization. We prove a more general result (Theorem~\ref{th:dense_network_gen2}) in appendix~\ref{subsec:proof:dense_network_gen}. 
	
\begin{theorem}\label{th:dense_network_gen1}
	Consider $\Phi$ a homogenous \ac{PPP} on a closed set ${\cal M} \subset \RR^2$ with intensity $\lambda$, and $h$ given by assumptions~\ref{a:power_law_path}. Then, for all $N \geq 1$, $(z_1,\cdots,z_N) \in (M \setminus {\cal M} )^N$, the distribution of $({\cal I}(z_1), \cdots ,  {\cal I}(z_N))$ converges in distribution to a multivariate normal distribution when $\lambda \to +\infty$. Namely:
		\als{
			\frac{1}{\sqrt{\lambda}} & ({\cal I}(z_1) - \expec{{\cal I}(z_1)}, \cdots ,  {\cal I}(z_N) - \expec{{\cal I}(z_N)}) \sk
			&\to  {\cal N}(0,\Sigma(z_1, \cdots ,z_N)), 
		}
	with $\Sigma(z_1, \cdots , z_N )$ its covariance matrix. The covariance is given by:
	\als{
		\cov( & {\cal I}(z_i) / \sqrt{\lambda} , {\cal I}(z_{i^\prime})/ \sqrt{\lambda} ) \sk 
		&= \expec{G^2} P^2 \int_{ {\cal M} } ( \norm{z_{i^\prime} - x}  \norm{z_i - x})^{-\eta} dx.
	}
\end{theorem}
	Specializing theorem~\ref{th:dense_network_gen1} to $N=1$ and $N=2$ we obtain the limit distribution of ${\cal I}(z)$ and $({\cal I}(z) , {\cal I}(z^\prime) )$ which can be used to compute $K$ and $L$ respectively. We do not write the $N=2$ case for clarity.
\begin{corollary}\label{cor:dense_network_1}
	With assumptions~\ref{a:power_law_path}, conditional to $z \in {\cal C}(0)$, the interference at $z$ converges in law to a normal distribution:
	\eqs{\frac{1}{\sqrt{\lambda}}( {\cal I}(z)  - \expec{ {\cal I}(z)| z} ) \tends{\lambda}{+\infty}  ( 0, \Sigma(z)) , }
	 with:
	\eqs{\var( {\cal I}(z)/ \sqrt{\lambda} | z) =  \pi \expec{G^2} P^2  \abs{z}^{2(1 - \eta)}/(\eta - 1) .}
\end{corollary}

%% file: numerical_exp.tex
\section{Numerical experiments}\label{sec:numerical_experiments}

	In this section we perform numerical experiments to study the distribution of the loads. We consider elastic traffic, so that the function $f$ is given by equations~\eqref{eq:elastic_loads} and~\eqref{eq:elastic_loads_data_rate}.
	
	We use the following parameters: path-loss $h(r) = P r^{-\eta}$ with $P = 10^{-128/10 - 3 \eta}$, $\eta > 2$, no shadowing $G \equiv 1$, total bandwidth $w = 5$ Mhz, total \ac{BS} transmit power $46$ dBm, thermal noise power $N_0 = -174$ dBm/Hz, traffic  $\sigma \lambda_{us} = 10$ Mbits/km${}^2$/s. We use $\eta = 3.5$ unless otherwise stated, which corresponds to typical dense urban environments. To simulate the distribution of the loads we proceed similarly to~\cite{GloaguenFSS05} to simulate the typical cell. To simulate the interference received over the typical cell, we simulate the \ac{PPP} on a window large enough so that there are $3 \times 10^{3}$ interfering \acp{BS} on average. To evaluate the distribution of the loads, $10^{4}$ independent samples of the typical cell are drawn. Numerical integrals are computed using Gauss-Kronrod quadrature.
\subsection{Loads distribution}
\begin{figure}
	  \centering
	  \begin{subfigure}[b]{\figsize}
	          \centering
	          \includegraphics[width=\textwidth]{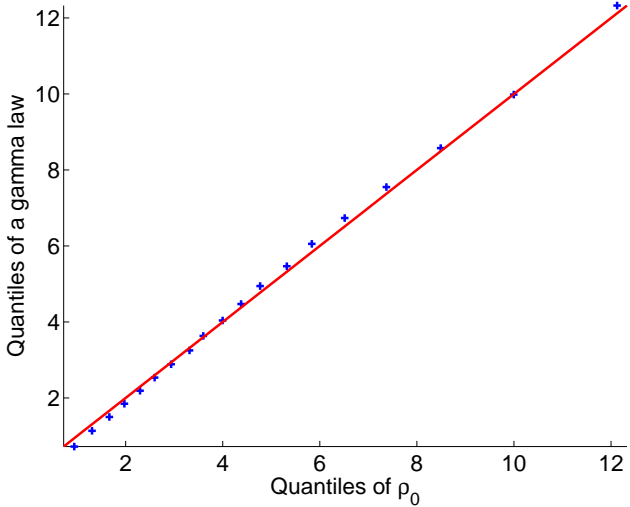}
	          \caption{$\eta=2.5$, $\lambda = 1 \text{BS/km}^2$}
	          \label{subfig:compare_loads1}
	  \end{subfigure}\\
	  \begin{subfigure}[b]{\figsize}
	          \centering
	          \includegraphics[width=\textwidth]{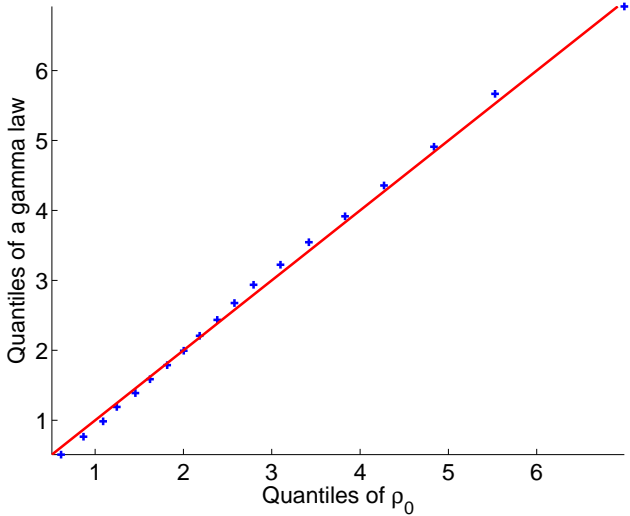}
	          \caption{$\eta=3.2$, $\lambda = 1 \text{BS/km}^2$}
	          \label{subfig:compare_loads2}
	  \end{subfigure}\\
	  \begin{subfigure}[b]{\figsize}
	          \centering
	          \includegraphics[width=\textwidth]{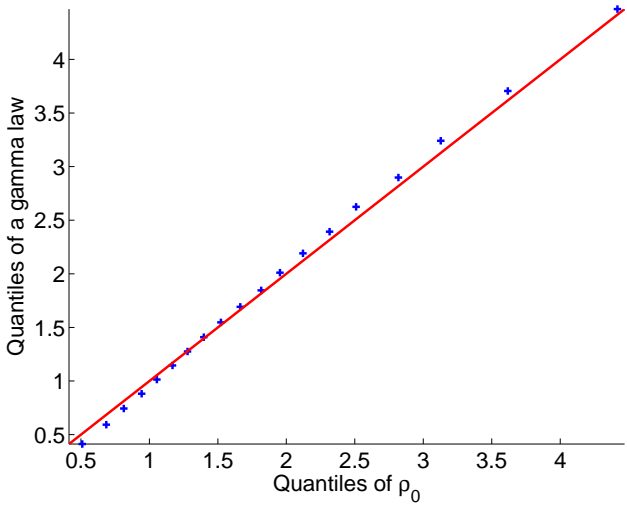}
	          \caption{$\eta=4$, $\lambda = 1 \text{BS/km}^2$}
	          \label{subfig:compare_loads3}
	  \end{subfigure}
	  \caption{Quantile-to-Quantile comparison between the loads distribution and a gamma distribution}\label{fig:gamma_loads}
\end{figure}
	On figure~\ref{fig:gamma_loads} we compare the distribution of the loads obtained by simulation to a gamma distribution with the same mean and variance for $\eta \in \Set{2.5, 3.2, 4}$ using a quantile-to-quantile plot. For the three values of $\eta$ we observe a good fit between the two. It is noted that from the results of~\cite{Calka2003}, even the area of the Voronoi cell \emph{does not} follow a gamma distribution, so in general the load cannot be gamma distributed. For practical purposes however, the gamma distribution provides a good approximation.
	
	It is noted that choosing a gamma distribution with the same mean and variance is not the maximum likelihood estimation, since the maximum likelihood involves the expectation of the logarithm of the load. We use this moment estimator because we only know the mean and variance of the loads by Theorems~\ref{th:plane_load_mean} and~\ref{th:plane_load_var}. We recall that from corollary~\ref{cor:th:plane_load_var} the load distribution becomes heavy-tailed ($\expecp{\rho_0} \to +\infty$) when $\eta \to 2^+$, so that it cannot be approximated by a light-tailed distribution such as the gamma in this case.

\subsection{Average loads}
On figure~\ref{fig:frequency_reuse}, we plot the average load as a function of the \ac{BS} density $\lambda$ with $\eta \in \Set{2.5, 3.5, 4}$. The four curves are denoted: 
\begin{itemize}
\item a) ``Average load'': the average load calculated by~\eqref{eq:plane_load_mean} of theorem~\ref{th:plane_load_mean},
\item b) ``$I \approx E[I]$'': an approximation of the average load by replacing the interference at a given location ${\cal I}(z)$ by its expectation $\expec{ {\cal I}(z)|z}$  .
\item c) ``$I - E[I]\approx N(0,\Sigma)$'': an approximation of the average load by replacing the interference at a given location ${\cal I}(z)$ by a normally distributed \ac{r.v.} with variance given by corollary~\ref{cor:dense_network_1}.
\item d) ``$I \equiv 0$'': the average load when there is no interference, i.e ${\cal I}(z) \equiv 0$.
\end{itemize}

	When $\lambda$ increases, the interference increases, but the size of the typical cell decreases so that the amount of traffic served by a given \ac{BS} decreases. Namely there is more interference but less congestion. On all curves, the average load decreases with $\lambda$ which shows that the effect of reduced congestion dominates the effect of increased interference.
	
	We observe that approximations b) and c) are fairly accurate, and that the accuracy is better when $\eta$ is close to $2$. This is because, when $\eta$ decreases, the ratio between interference and path loss $\frac{ \expec{ {\cal I}(z) | z } }{h(z)}$ increases, we get closer to the low \ac{SINR} regime and $I \to f(z,I)$ becomes closer to a linear function as explained in subsection~\ref{subsec:influence_interference}. The same argument also justifies the fact that for a given value of $\lambda$, the average load increases when the path-loss exponent $\eta$ decreases.

	There is a large difference between the load with interference a) and without d), showing that taking into account interference for predicting the loads distribution is indeed necessary.
\begin{figure}
	  \centering
	  \begin{subfigure}[b]{\figsize}
	          \centering
	          \includegraphics[width=\textwidth]{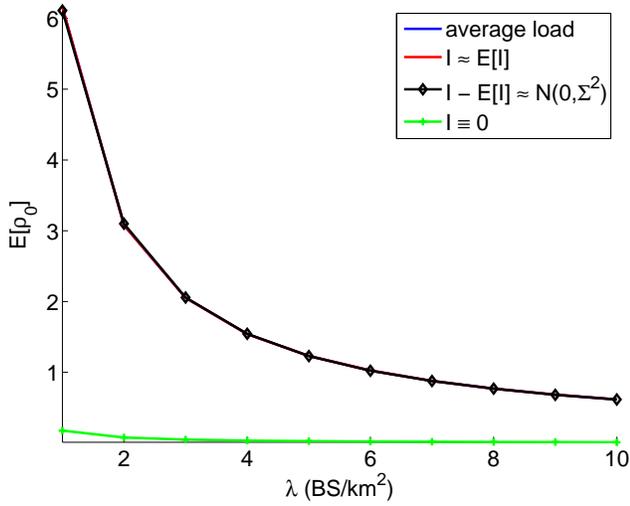}
	          \caption{$\eta=2.5$}
	          \label{subfig:average_loads1}
	  \end{subfigure}\\
	  \begin{subfigure}[b]{\figsize}
	          \centering
	          \includegraphics[width=\textwidth]{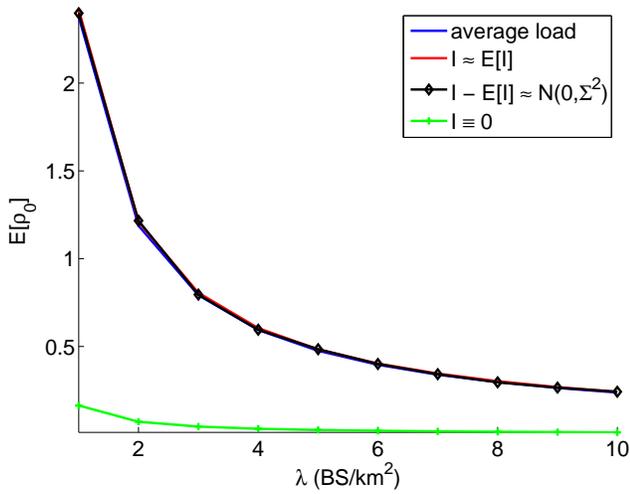}
	          \caption{$\eta=3.5$}
	          \label{subfig:average_loads2}
	  \end{subfigure}\\
	  \begin{subfigure}[b]{\figsize}
	          \centering
	          \includegraphics[width=\textwidth]{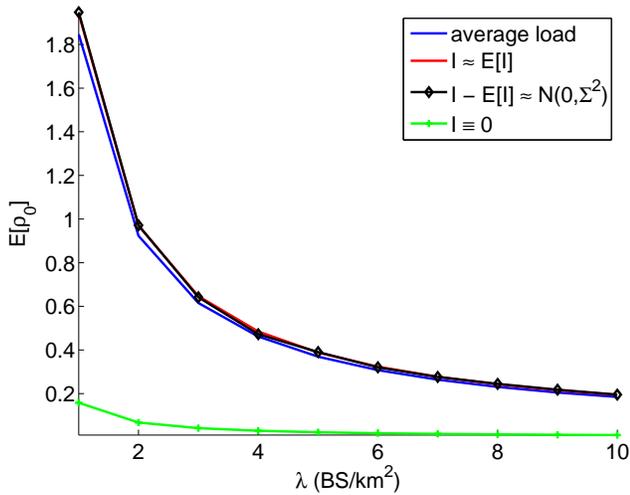}
	          \caption{$\eta=4$}
	          \label{subfig:average_loads3}
	  \end{subfigure}
	  \caption{Average load as a function of the BS density $\lambda$}\label{fig:average_loads}
\end{figure}
\subsection{Frequency reuse}
	On figure~\ref{fig:frequency_reuse}, we plot the average load as a function of $\lambda$ for different frequency reuse strategies as explained in subsection~\ref{subsec:sfr}. On figure~\ref{subfig:frequency_reuse1} we consider hard reuse with reuse factor $b \in \Set{1,\dots,4}$.  On figure~\ref{subfig:frequency_reuse2} we consider soft reuse with $b = 3$ and $\kappa \in [-40,0]$ dB. For each value of $\lambda$ and $\kappa$, the threshold between edge and center $r_{edge}$ is chosen so that the average load of the center and edge are equal: $\expecp{\rho_0^{edge}} = \expecp{\rho_0^{center}}$. On figure~\ref{subfig:frequency_reuse3} we consider a soft reuse with $b = 3$, $\lambda = 1 \text{BS/km}^2$ and $\kappa \in [-40,0]$ dB.
	
	We see that hard reuse schemes actually diminish the network performance by \emph{increasing} the average load. Reuse $1$ gives the best performance. This is true regardless of the \ac{BS} density $\lambda$. On the other hand, soft reuse schemes bring considerable improvement. The value $\kappa = -20$ dB gives the best performance, and the average load for this value is about half of the average load for reuse $1$. This shows that for \ac{PPP} networks, soft frequency reuse allows to serve twice as much traffic with the same number of \acp{BS}.
	
	  Those results are in line with previous work such as~\cite{BonaldHegdeFrequencyReuse} where the capacity of a regular network (possibly with small random perturbations) serving elastic traffic was studied. Namely soft reuse increases the network performance noticeably, while hard reuse brings little to no improvement. A noticeable difference though is that the optimal power reduction for the completely random \ac{PPP} networks $\kappa$ is $-20$ dB, while for regular networks it is around $-5$ dB. Namely there is more than an order of magnitude of difference. 

\begin{figure}
	  \centering
	  \begin{subfigure}[b]{\figsize}
	          \centering
	          \includegraphics[width=\textwidth]{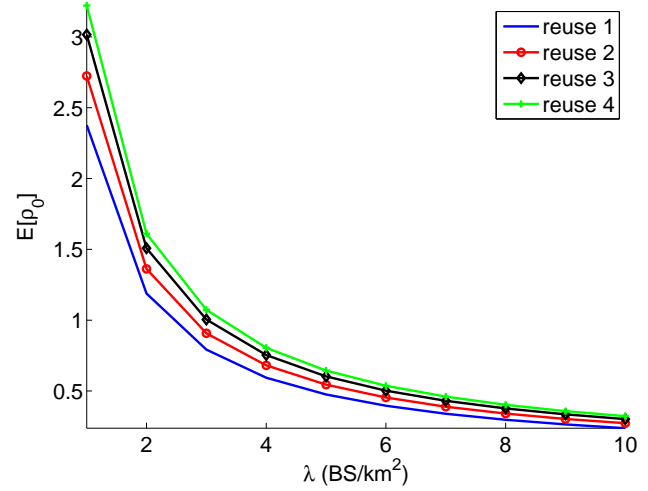}
	          \caption{Hard reuse, $\eta=3.5$}
	          \label{subfig:frequency_reuse1}
	  \end{subfigure}\\
	  \begin{subfigure}[b]{\figsize}
	          \centering
	          \includegraphics[width=\textwidth]{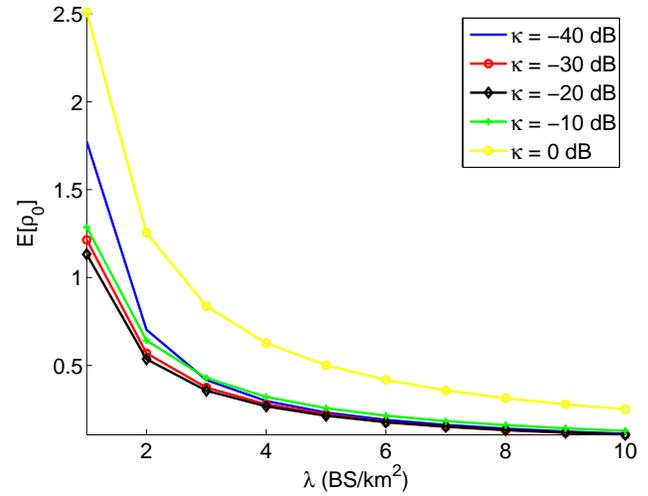}
	          \caption{Soft reuse $3$, $\eta=3.5$}
	          \label{subfig:frequency_reuse2}
	  \end{subfigure}\\
	  \begin{subfigure}[b]{\figsize}
	          \centering
	          \includegraphics[width=\textwidth]{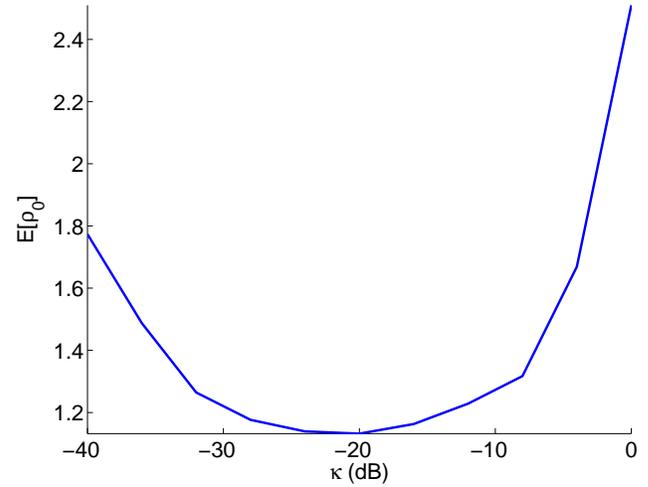}
	          \caption{Soft reuse $3$, $\eta=3.5$, $\lambda = 1 \text{BS/km}^2$: \\ \centering{ optimal $\kappa$} }
	          \label{subfig:frequency_reuse3}
	  \end{subfigure}	  
	  \caption{Average load for different frequency reuse strategies}\label{fig:frequency_reuse}
\end{figure}

%% file: conclusion.tex
\section{Conclusion}\label{sec:conclusion}	
	We have considered the flow-level performance of random wireless networks. The locations of \acp{BS} follow a \ac{PPP}. We take into account flow-level dynamics  by modeling each \ac{BS} as a queue. The performance and stability of each \ac{BS} depends on its load. In certain cases the full distribution of the load has been derived through its Laplace transform. In all cases, we have given formulas to calculate the first and second moments of the loads. Networks on the line and on the plane have both been considered. Our model is generic enough to include features of recent wireless networks such as 4G networks (LTE). In dense networks, we have shown that the inter-cell interference power becomes normally distributed. Hence some computations reduce to an integral with respect to the Gaussian distribution. Using numerical experiments we have demonstrated that in certain cases of practical interest, the loads distribution can be approximated by a gamma distribution with known mean and variance, and that this approximation is very accurate.
	
\section*{Acknowledgment}
	The authors are thankful to Alexandre Prouti\`ere for his helpful comments on the traffic models.

%% file: proofs.tex
\section{Auxiliary results on shot noise generated by a PPP}\label{sec:shotnoise}
	\subsection{Laplace transform of shot noise}
	The Laplace transform of shot noise is given by Theorem~\ref{prop:shot_noise_laplace}. A proof can be found for instance in \cite{Baccelli2009}.
	\begin{proposition}\label{prop:shot_noise_laplace}
		Consider $\Phi$ a \ac{PPP} on $M$ with measure $\lambda(dz)$. for all $N \geq 1$, $(z_1,\cdots,z_N) \in M^N$, $(s_1,\cdots,s_N) \in \CC^N$: with strictly positive real part:
		\als{
			\expec{  \exp(  - \sum_{i=1}^N  s_i	{\cal I}(z_i))  }  \sk 
			= \exp \left\{ - \int_{M} \left[ 1 - {\cal G}\left( s_i  \sum_{i=1}^N h(\norm{z_i - z}) \right) \right] \lambda(dz) \right\}.
		}
	\end{proposition}
		\subsection{Expectation of a function of shot noise}
	It is noted that once the Fourier/Laplace transform of a random variable $I$ is known, the expectation of any function of $I$ can be calculated using the Plancherel-Parseval theorem recalled in proposition~\ref{prop:parseval}.
\begin{proposition}[Plancherel theorem] \label{prop:parseval}
	Consider $I$ a real random variable and $f$ a positive function:
	\eqs{
		\expec{f(I)} = \int_{\RR} \overline{\hat{f}(s)} \expec{ \exp( -s i I )} ds.
	}
	with $\hat{f}$ the Fourier transform of $f$, and $\overline{z}$ the complex conjugate of $z$.
\end{proposition}	
	It is noted that if $f$ is not square integrable, the Fourier transform should be taken in the sense of tempered distributions.
	\subsection{Shot noise distributions of interest}
	We can specialize Theorem~\ref{prop:shot_noise_laplace} to obtain results on the distribution of shot noise generated by an \emph{homogenous} \ac{PPP}, conditional to the event $z \in {\cal C}(0)$.
	\begin{proposition}\label{prop:shot_noise_plane}
	The Laplace transform of ${\cal I}(z)$, conditional to $z \in {\cal C}(0)$ is:
	\eqs{ 
	\expec{ \exp( - s {\cal I}(z) )|  z } = \exp \Lp - 2 \pi \lambda \int_{\abs{z}}^{+\infty} r( 1 - {\cal G}(s h(r)))  dr \Rp.
	}
	and its first and second moments are:
	\als{
	\expec{ {\cal I}(z) |  z } &= 2 \pi \lambda \expec{G}  \int_{\abs{z}}^{+\infty} r h(r) dr,\sk
	\var( {\cal I}(z)  |  z ) &=  2 \pi \lambda \expec{G^2} \int_{\abs{z}}^{+\infty} r h(r)^2 dr,
	}
	\end{proposition}
	\begin{proof}
		By specialization of theorem~\ref{prop:shot_noise_laplace} to $N = 1$ and $\lambda(dx) = \lambda \indic_{ \CC \setminus {\cal B}_1(z)}(x) dx$:
	\als{
	 	 & \expec{ \exp( - s {\cal I}(z) )|  z }) \sk &= \exp \Lp -\lambda \int_{\CC \setminus {\cal B}_1(z)}(1-{\cal G}(s h(z - x)) ) dx \Rp.
	}
	Once again, using polar coordinates and centering (i.e $z - x = r e^{i \theta}$) gives the result:
	\eqs{
	\expec{ \exp( - s {\cal I}(z) )|  z } = \exp \Lp - 2 \pi \lambda \int_{\abs{z}}^{+\infty} ( 1 - {\cal G}(s h(r))) r dr \Rp.
	}
	\end{proof}
	
\begin{proposition}\label{prop:shot_noise_plane2}
	With assumptions~\ref{a:power_law_path}, the Laplace transform of ${\cal I}(z)$, conditional to $z \in {\cal C}(0)$ is given by:
	\als{ 
	 & \log( \expec{ \exp(  s {\cal I}(z) )|  z } ) \sk 
	 &=   \pi \lambda \Lb \abs{z}^2 \Rd +  \Ld  2 (P s)^{\frac{2}{\eta}} \gamma(P s \abs{z}^{-\eta},-2/\eta)/\eta \Rb ,
	}
	with $\gamma(x,s) = \int_{0}^{x} t^{s-1} e^{-t} dt$ the lower incomplete gamma function and $\Gamma(s) = \Gamma(+\infty,s) $ the gamma function.	Its first and second moments are:
	\als{
	\expec{ {\cal I}(z) |  z } &= 2 \pi \lambda \expec{G} P  r^{2-\eta}/(\eta-2) ,\sk
	\var( {\cal I}(z)  |  z ) &=  2 \pi \lambda \expec{G^2} P^2 r^{2-2\eta}/(2\eta-2). 
	}
\end{proposition}
\begin{proof} See \cite{AndrewsBaccelliGeometry}. \end{proof}
\section{Proofs}\label{sec:proofs}
\subsection{Proof of theorem~\ref{th:load_line_no_interf}}\label{subsec:proof_th_load_line_no_interf}
\begin{proof}
	 First assume $(x_l,x_r)$ to be known. We decompose the load~\eqref{eq:line_load} in two terms using~\eqref{eq:line_load_distance} and the symmetry of $f_0$:
	\als{
		\rho_0 &= \int_{-x_l/2}^{x_r/2} f_0(z) dz = \int_{0}^{x_r/2} f_0(z) dz + \int_{0}^{x_l/2} f_0(z) dz \sk 
		&= F_0(x_l/2) + F_0(x_r/2).
	}
	By proposition~\ref{prop:geometry_line_cell}, $(x_l,x_r)$ are \ac{i.i.d.} so that:
	\al{\label{eq:th1e2}
	\expecp{ \exp( -s \rho_0 )} &= \expecp{ \exp( -s (F_0(x_l/2) + F_0(x_r/2)) )} \sk 
	&= ( \expecp{ \exp( -s F_0(x_l/2)) } )^2.
	}
	$x_l$ is exponentially distributed with parameter $\lambda$ so that the \ac{r.h.s.} of~\eqref{eq:th1e2} equals:
	\als{
		& \expecp{ \exp( -s \rho_0 )} \sk
		&=   \left(  \lambda \int_{0}^{+\infty} \exp( -s F_0(r/2)) \exp(-r \lambda) dr  \right)^2,
	}
	proving the result.
\end{proof}

\subsection{Proof of theorem~\ref{th:linear_interf_plane}}\label{subsec:proof_th_linear}
\begin{proof}
	 We first consider $(x_l,x_r)$ to be known. The interference at $z \in {\cal C}(0)$ can be decomposed into:
	\als{
		{\cal I}(z) &= G_{-1} h(\abs{z + x_l}) + G_1 h( \abs{z - x_r}) \sk
		&+ \sum_{n \in \NN \setminus \Set{-1,0,1} } G_n h(\abs{z - x_n}).
	}
	The interference dependent term in~\eqref{eq:load_line_linear_interf} can be written:
	\al{ \label{eq:th1eq2}
		\int_{-x_l/2}^{x_r/2} f_1(z) & {\cal I}(z) dz = G_{-1} \int_{-x_l/2}^{x_r/2} f_1(z) h(\abs{z + x_l})dz \sk
		&+ G_1 \int_{-x_l/2}^{x_r/2} f_1(z)  h( \abs{z - x_r})dz \sk
		&+  \sum_{n \in \NN \setminus \Set{-1,0,1} } G_n \int_{-x_l/2}^{x_r/2} f_1(z) h(\abs{z - x_n}) dz.
	}
	By definition of ${\cal K}$ (eq.~\ref{eq:th_1_def_k}), the \ac{r.h.s.} of \eqref{eq:th1eq2} becomes:
	\al{\label{eq:th1eq3}
	\int_{-x_l/2}^{x_r/2} f_1(z) {\cal I}(z) dz &= G_{-1} {\cal K}(-x_l,x_l,x_r) + G_{1} {\cal K}(x_r,x_l,x_r) \sk
	&+ \sum_{n \in \NN \setminus \Set{-1,0,1} }   G_{n} {\cal K}(x_n,r_l,r_r).
	}
	The third term on the \ac{r.h.s.} of \eqref{eq:th1eq3} is a shot noise with respect to the point process $\Phi \setminus [-x_l, x_r]$, with impulse response $z \mapsto {\cal K}(z,x_l,x_r)$. Conditional to $(x_l,x_r)$, $\Phi \setminus [-x_l, x_r]$ is a \ac{PPP} on $\RR \setminus [-x_l,x_r]$ with intensity $\lambda$. Hence we can apply proposition~\ref{prop:shot_noise_laplace}:
	\al{\label{eq:th1eq4}
	 \expecp{ \exp\left\{ -s \sum_{n \in \NN \setminus \Set{-1,0,1} } G_{n} {\cal K}(x_n,r_l,r_r)  \right\} | x_l,x_r } \sk
	 = \exp \left\{ -\lambda \int_{ \RR \setminus [-r_l,r_r] } (1 -  {\cal G}( s {\cal K} (z,r_l,r_r) ) ) dz \right\}
	}
	Furthermore, the Laplace transform of the two first terms of~\eqref{eq:th1eq3} is:
\al{\label{eq:th1eq5}
& \expecp{ \exp( -s G_{-1} {\cal K}(-x_l,x_l,x_r) -s G_{1} {\cal K}(x_r,x_l,x_r) ) | x_l,x_r} \sk
& = {\cal G}( s {\cal K}(-x_l,x_l,x_r) ) {\cal G}( s {\cal K}(x_r,x_l,x_r) ).
}
	Combining~\eqref{eq:th1eq4} and~\eqref{eq:th1eq5} we obtain the Laplace transform of the interference dependent term in~\eqref{eq:load_line_linear_interf}:
\als{
&\expecp{ \exp \Lp -s \int_{-x_l/2}^{x_r/2} f_1(z) {\cal I}(z) dz \Rp | x_l,x_r} \sk
&=  {\cal G}( s {\cal K}(-x_l,x_l,x_r) ) {\cal G}( s {\cal K}(x_r,x_l,x_r) ) \sk
& \exp \Lb -\lambda \int_{ \RR \setminus [-x_l,x_r] } (1 -  {\cal G}( s{\cal K} (z,x_l,x_r) ) ) dz \Rb \sk
&=  {\cal L}(s, x_l , x_r).
}
	The Laplace transform of the load conditional to $(x_l,x_r)$ is then:
\al{\label{eq:th1eq7}
	\expecp{ \exp( -s \rho_0 ) | x_l,x_r} \sk
	= {\cal L}(s, x_l , x_r) \exp( -s ( F_0(x_l/2) + F_0(x_r/2) )).
}
	Since $(x_l,x_r)$ are \ac{i.i.d.} and exponentially distributed, the load distribution is obtained by un-conditioning~\eqref{eq:th1eq7}:
	\al{\label{eq:th1eq8}
	& \expecp{ \exp( -s  \rho_0 )} = \lambda^2 \int_{\RR^2} {\cal L}(s,r_1 , r_2) \sk
	& \exp \left\{ -\lambda ( r_1  + r_2 ) - s ( F_0(r_1/2)  + F_0(r_2/2)) \right\} d r_1 d r_2.
	}
	The expression under the integral in~\eqref{eq:th1eq8} is symmetrical in $(r_1,r_2)$ which proves the result.
\end{proof}

\subsection{Proof of Theorem~\ref{th:plane_load_moments}}

\begin{proof}
	We first write the definition of the load:
	\als{ 
		\rho_0^N &= \left(  \int_{\CC} f(z,{\cal I}(z))  \indic_{ {\cal C}(0)}(z) dz  \right)^N \sk
		&= \int_{\CC^N}  \left[ \prod_{i=1}^N f(z_i, {\cal I}(z_i)) \right] \indic_{ {\cal C}(0)^N}(z_1,\cdots,z_N) dz_1 \cdots dz_N.
	}
	Taking expectations:
	\als{ 
		&\expecp{\rho_0^N} \sk
		&= \int_{\CC^N} \expecp{ \left[ \prod_{i=1}^N f(z_i, {\cal I}(z_i)) \right] \indic_{ {\cal C}(0)^N}(z_1,\cdots,z_N)  } dz_1 \cdots dz_N.
	}
	The probability of belonging to ${\cal C}(0)^N$ is given as a void probability of the \ac{PPP}:       	
	\als{ 
	\expecp{ 1_{ {\cal C}(0)^N }(z_1,\cdots,z_N)} &= \probap{ \Phi \cap {\cal B}_N(z_1,\cdots,z_N) = \Set{0}} \sk
	&= \exp(-\lambda  B_N(z_1,\cdots,z_N)),
	}
	and by conditioning on $(z_1,\cdots,z_N) \in {\cal C}(0)^N$, we obtain the first result:
	\als{ 
		 \expecp{\rho_0^N} = \int_{\CC^N} & \expecp{ \prod_{i=1}^N f(z_i, {\cal I}(z_i)) | z_1,\cdots,z_N  } \sk  
		 &\exp(-\lambda B_N(z_1,\cdots,z_N) ) dz_1 \cdots dz_N.
	}
	
	Conditional to $(z_1,\cdots,z_N) \in {\cal C}(0)^N$, $\Phi \setminus \Set{0}$ is a \ac{PPP} on $\CC \setminus {\cal B}_N(z_1,\cdots,z_N)$ with intensity $\lambda$. Hence applying proposition~\ref{prop:shot_noise_laplace} proves the second result.
\end{proof}

\subsection{Proof of Theorem~\ref{th:plane_load_var}}\label{subsec:proof:plane_load_var}
\begin{proof}
 	a) Specialization of theorem~\ref{th:plane_load_mean} to $N=2$ gives:
 	\eq{\label{eq:th4eq1} 
 	 \expecp{ \rho_0^2 } = \int_{\CC^2} L(z,z^\prime) \exp( - \lambda B_2(z,z^\prime)) dz dz^\prime.
 	}
 	In polar coordinates~\eqref{eq:th4eq1} becomes:
 	\al{\label{eq:th4eq2} 
 		\expecp{ \rho_0^2 } &= \int_{(\RR^+)^2} \int_{[0,2 \pi]^2} \sk
 		&  r r^\prime L(r e^{i\theta},r^\prime e^{i\theta^\prime}) \exp( - \lambda B_2(r e^{i\theta},r^\prime e^{i\theta^\prime}) )dr dr^\prime d\theta d\theta^\prime.
 	}
 	The integrand in~\eqref{eq:th4eq2} has the following invariants:
 	\begin{itemize}
 	\item $(r e^{i\theta},r^\prime e^{i\theta^\prime}) \to (r^\prime e^{i\theta^\prime},r e^{i\theta})$  (symmetry)
 	\item $(\theta,\theta^\prime) \to (\theta - \theta^\prime , 0)$  (rotation invariance)
 	\item $B_2(r e^{i\theta},r^\prime e^{i\theta^\prime})  = r^2 B_2(e^{i\theta},r^\prime/r e^{i\theta^\prime}) = {r^\prime}^2 B_2(1, r/ {r^\prime} e^{i(\theta-\theta^\prime) })$, (homogeneity of $B_2$)
 	\end{itemize}
 	which gives the announced formula.
	
\end{proof}

\subsection{Proof of Corollary~\ref{cor:th:plane_load_var}}\label{subsec:cor:th:plane_load_var}

\begin{proof}
	By concavity of the logarithm:
	\eqs{
		f(z, {\cal I}(z)) \geq \frac{\log(2) \lambda_{us} \sigma}{w} \frac{ N_0 +  {\cal I}(z) }{ h(z) }.
	}
	Applying corollary~\ref{cor:plane_load_mean_linear_path}, with $f_0(z) = \frac{\lambda_{us} \sigma \log(2) N_0 }{ w h(z) } $ and $f_1(z) = \frac{\lambda_{us} \sigma \log(2)}{ w h(z) }$:
	\als{
	\expecp{ \rho_0 } & \geq 2 \pi \int_{0}^{+\infty} r \Lp f_0(r) +  2 \pi \lambda P \frac{r^{2-\eta}}{\eta - 2} f_1(r) \Rp e^{ - \lambda \pi r^2 } dr.
	}
	We have that:
	\als{
	& 2 \pi \int_{0}^{+\infty} r f_0(r) e^{ - \lambda \pi r^2 } dr \sk 
	&= \frac{2 \pi \log(2) \lambda_{us} \sigma  N_0}{w P}  \int_{0}^{+\infty} r^{\eta + 1} e^{ - \lambda \pi r^2 } dr \sk
	&= \frac{\log(2) \lambda_{us} \sigma}{w \lambda}  \frac{N_0  \Gamma( 1 + \eta/2)}{ P ( \sqrt{\pi \lambda} )^{\eta}},
	}
	and:
	\als{
	& 2 \pi \int_{0}^{+\infty} r 2 \pi \lambda P \frac{r^{2-\eta}}{\eta - 2} f_1(r) e^{ - \lambda \pi r^2 } dr \sk
	&= \frac{(2 \pi)^2 \log(2) \lambda_{us} \sigma \lambda}{(\eta - 2)w}  \int_{0}^{+\infty} r^{3} e^{ - \lambda \pi r^2 } dr \sk
	&=  \frac{\log(2) \lambda_{us} \sigma}{w \lambda}  \frac{2}{\eta - 2}. 
	}
	We have used the following identity twice (by a change of variables):
	\als{
		\int_{0}^{+\infty} r^a e^{ - \lambda \pi r^2 } dr &= \frac{\Gamma((a + 1)/2 )}{ 2 (\lambda \pi)^{(a + 1)/2} } , 
	}
	for $a \geq 0$. We obtain the announced result by summing:
	\als{
	\expecp{ \rho_0 } & \geq \frac{ \log(2) \lambda_{us} \sigma}{w \lambda} \Lp \frac{N_0 \Gamma( 1 + \eta/2)}{ P ( \sqrt{\pi \lambda} )^{\eta}}   +  \frac{2}{\eta - 2} \Rp.
	}
\end{proof}

\subsection{Proof of Theorem~\ref{th:dense_network_gen1}}\label{subsec:proof:dense_network_gen}
	Instead of proving Theorem~\ref{th:dense_network_gen1}, we prove a more general result, Theorem~\ref{th:dense_network_gen2}. We introduce assumption~\ref{a:dense_network} on the decay of the path loss function $h$. In particular, with assumption~\ref{a:power_law_path}, assumption~\ref{a:dense_network} holds.
\begin{assumptions}\label{a:dense_network}
	There exists $\epsilon > 0$, $z_0 \in \CC$ and a constant $C_h > 0$ such that $h(z) \leq  C_h \abs{z}^{-(2+\epsilon)}$ for $\abs{z} \geq \abs{z_0}$.
\end{assumptions}
\begin{theorem}\label{th:dense_network_gen2}
	Consider $\Phi$ a \ac{PPP} on $M$ with measure $\lambda m(dz)$ where $\lambda \in \RR^+$ and $m$ is a measure on $M$ such that the following integrals are finite for all $z \in M$:
\begin{itemize}
	\item $\int_{M} h(\norm{z - x}) m(dx) 	< +\infty$,
	\item $\int_{M} h(\norm{z - x})^2 m(dx) < +\infty$.
\end{itemize}
	Furthermore we assume that for all $z$, $\esssup_{x} h(\norm{z - x}) <+\infty$ where $\esssup$ is taken with respect to measure $m$ and that $\expec{G^2} < +\infty$.
	
	For all $N \geq 1$, $(z_1,\cdots,z_N) \in M^N$, $({\cal I}(z_1), \cdots ,  {\cal I}(z_N))$ converges in distribution to a multivariate normal distribution when $\lambda \to +\infty$. Namely:
		\als{
			\frac{1}{\sqrt{\lambda}} & ({\cal I}(z_1) - \expec{{\cal I}(z_1)}, \cdots ,  {\cal I}(z_N) - \expec{{\cal I}(z_N)}) \sk
			&\to  {\cal N}(0,\Sigma(z_1, \cdots ,z_N)), 
		}
	with $\Sigma(z_1, \cdots , z_N )$ its covariance matrix. The covariance is given by:
	\als{
		\cov( & {\cal I}(z_i) / \sqrt{\lambda} , {\cal I}(z_{i^\prime})/ \sqrt{\lambda} ) \sk 
		&= \expec{G^2} \int_{M} h(\norm{z_{i^\prime} - x})  h(\norm{z_i - x}) m(dx).
	}
\end{theorem}

\begin{proof}
	We recall the Laplace transform of the shot noise from proposition~\ref{prop:shot_noise_laplace}:
	\al{\label{eq:th5eq1}
			&\expec{  \exp(  - \sum_{i=1}^N  s_i	{\cal I}(z_i))  } \sk
			&= \exp \left\{ - \int_{M} \left[ 1 - {\cal G}\left( s_i  \sum_{i=1}^N h(\norm{z_i - z}) \right) \right] \lambda m(dz) \right\}.
	}
	The expectation of shot noise is given by:
	\eq{ \label{eq:th5eq2}
		\expec{  {\cal I}(z_i) } =  \expec{G} \lambda \int_{M} h(\norm{z_i - x}) m(dx).
	}
	Replacing \eqref{eq:th5eq2} in \eqref{eq:th5eq1} we have that:
	\als{
	& \expec{  \exp(  - \sum_{i=1}^N  \frac{s_i}{\sqrt{\lambda}}	( {\cal I}(z_i) - \expec{ {\cal I}(z_i)}  ) ) } \sk
	&= \exp \Lb - \lambda \int_{M} \left[ 1 - {\cal G} \Lp \sum_{i=1}^N \frac{s_i h(\norm{z_i - z})}{\sqrt{\lambda}}  \Rp \Rd \Rd \sk
	  & \Ld \Ld + \expec{G} \sum_{i=1}^N \frac{s_i}{\sqrt{\lambda}} h(\norm{z_i - z})  \right] m(dz) \Rb.
	}
	Consider $s$ fixed and consider $\delta > 0$. Since $\esssup_x h(\norm{z - x}) <+\infty$, there exists $\lambda$ so that $m$-almost everywhere: 
	\eqs{
		\frac{\abs{ s h(\norm{z - x}) }}{\sqrt{ \lambda}}  < \delta.
	}
	Using a Taylor expansion for ${\cal G}$ in a neighborhood of $0$, there exists a function $a$ such that:
	\eqs{
		{\cal G}(\delta) = 1 -  \expec{G} \delta + \expec{G^2} \frac{\delta^2}{2} + \abs{\delta}^2 a(\delta ).
	}
	with $a(\delta) \to 0$ , $\delta \to 0^+$. Therefore:
	\al{\label{eq:th5eq6}	
	& \expec{  \exp(  - \sum_{i=1}^N  \frac{s_i}{\sqrt{\lambda}}	( {\cal I}(z_i) - \expec{ {\cal I}(z_i)}  ) )  } \sk
	&= \exp \left\{ a(\delta) - \frac{\expec{G^2}}{2} \int_{M} \left[   \sum_{i=1}^N s_i h(\norm{z_i - z}) \right]^2 m(dz) \right\} \sk
	 &= \exp \left\{ a(\delta) -  \Rd \sk 
	 & \Ld  \sum_{1 \leq i,i^\prime \leq N} \frac{\expec{G^2} s_i s_{i^\prime}}{2}  \int_{M}  h(\norm{z_i - z}) h(\norm{z_{i^\prime} - z}) m(dz) \right\}
	}
	 It is noted that we can choose $\delta$ arbitrarily small when $\lambda$ is arbitrarily large. Therefore the \ac{r.h.s.} of equation ~\eqref{eq:th5eq6}	is, up to a negligible term, the Laplace transform of the multivariate Gaussian distribution which proves the convergence in distribution.
	
	The covariance is obtained by inspection of the Laplace transform \eqref{eq:th5eq6}.
\end{proof}